\documentclass[final]{dmtcs-episciences}

\usepackage{amsmath,amsthm,amssymb}
\usepackage{xspace}
\usepackage{graphicx}
\usepackage{xcolor}
\usepackage{algorithm}
\usepackage[noend]{algorithmic}
\usepackage{hyperref}
\usepackage{multirow}
\usepackage{tikz}
\usepackage{datetime}
\usepackage[utf8]{inputenc}
\usepackage{enumerate}
\usepackage{subcaption}
\usepackage{stmaryrd}
\usepackage[margin=30pt]{caption}

\newtheorem{theorem}{Theorem}

\newtheorem{corollary}[theorem]{Corollary}

\theoremstyle{definition}

\newcommand{\qedclaim}{\hfill $\diamond$ \medskip}

\DeclareMathOperator{\diam}{diam}

\author[Gagnon, Hassler, Huang, Krim-Yee, Mc Inerney, Mej\'ia Zacar\'ias, Seamone, Virgile]{\centering
Aliz\'ee Gagnon\affiliationmark{1} \and
Alexander Hassler\affiliationmark{2} \and
Jerry Huang\affiliationmark{3} \and
Aaron Krim-Yee\affiliationmark{4} \and
Fionn Mc Inerney\affiliationmark{5} \and
Andr\'es Mej\'ia Zacar\'ias\affiliationmark{6} \and
Ben Seamone\affiliationmark{1,7}\thanks{Corresponding author (\href{mailto:bseamone@dawsoncollege.qc.ca}{bseamone@dawsoncollege.qc.ca})} \and\\
Virg\'elot Virgile\affiliationmark{8}
}
\title[A method for eternally dominating strong grids]{A method for eternally dominating strong grids}
\affiliation{
D\'{e}partement d’informatique et de recherche op\'{e}rationnelle, Universit\'{e} de Montr\'{e}al, Montreal, QC, Canada \\
Facult\'e de m\'edecine et m\'edecine dentaire, Universit\'e catholique de Louvain Bruxelles-Woluwe, Brussels, Belgium \\
David Cheriton School of Computer Science, University of Waterloo, Waterloo, ON, Canada \\
Department of Bioengineering, McGill University, Montreal, QC, Canada \\
Universit\'e C\^ote d'Azur, Inria, CNRS, I3S, France \\
Instituto de Matem\'aticas, Universidad Nacional Aut\'onoma de M\'exico, M\'exico City, M\'exico \\
Mathematics Department, Dawson College, Montreal, QC, Canada \\
Department of Mathematics and Statistics, University of Victoria, Victoria, BC, Canada
}
\received{2019-2-4}

\revised{2020-2-11}

\accepted{2020-2-12}

\publicationdetails{22}{2020}{1}{8}{5162}

\begin{document}

\maketitle

\begin{abstract}
In the eternal domination game, an attacker attacks a vertex at each turn and a team of guards must move a guard to the attacked vertex to defend it. The guards may only move to adjacent vertices and no more than one guard may occupy a vertex. The goal is to determine the eternal domination number of a graph which is the minimum number of guards required to defend the graph against an infinite sequence of attacks.
In this paper, we continue the study of the eternal domination game on strong grids. Cartesian grids have been vastly studied with tight bounds for small grids such as $2\times n$, $3\times n$, $4\times n$, and $5\times n$ grids, and recently it was proven in [Lamprou et al., CIAC 2017, 393-404] that the eternal domination number of these grids in general is within $O(m+n)$ of their domination number which lower bounds the eternal domination number. Recently, Finbow et al. proved that the eternal domination number of strong grids is upper bounded by $\frac{mn}{6}+O(m+n)$. 
We adapt the techniques of [Lamprou et al., CIAC 2017, 393-404] to prove that the eternal domination number of strong grids is upper bounded by $\frac{mn}{7}+O(m+n)$. While this does not improve upon a recently announced bound of $\lceil\frac{m}{3}\rceil \lceil\frac{n}{3}\rceil+O(m\sqrt{n})$ [Mc Inerney, Nisse, P\'erennes, CIAC 2019] in the general case, we show that our bound is an improvement in the case where the smaller of the two dimensions is at most $6179$.
\end{abstract}

\section{Introduction}
\label{sec:intro}

\subsection{Background}

The graph security model of eternal domination was introduced in the 1990's with the study of the military strategy of Emperor Constantine for defending the Roman Empire in a mathematical setting \cite{AF95,Revelle97,RR00,Stewart99}.  The problem which is studied in these papers, roughly put, is how to defend a network of cities with a limited number of armies at your disposal in such a way that an army can always move to defend against an attack by invaders and do so for any sequence of attacks.  In the original version of eternal domination (also called ``infinite order domination'' \cite{BurgerCG+04} and ``eternal security''\cite{GHH05} in earlier works), $k$ guards are placed on the vertices of a graph $G$ so that they form a dominating set.  An infinite sequence of vertices is then revealed one at a time (called ``attacks'').  After each attack, a single guard is allowed to move to the attacked vertex.  If, after each attack, the guards maintain a dominating set, then we say that $k$ guards eternally dominate $G$.  The minimum $k$ for which $k$ guards can eternally guard $G$ for any sequence of attacks is called the eternal domination number of $G$, and is denoted $\gamma^{\infty}(G)$.

A subsequently introduced model, and the one we study here, allows any number of guards to move on their turn.  The minimum number of guards required to eternally dominate a graph $G$ in this model (called the ``all-guards move'' model) is denoted $\gamma^{\infty}_{all}(G)$, and is called the $m$-eternal domination number of $G$.  Typically, one requires that no two guards occupy the same vertex.  If one allows more one guard to occupy a vertex at a given time, then the corresponding parameter typically appears in the literature as $\gamma^{*\infty}_{all}(G)$; we do not consider this model here.  For more variants and a background on results related to eternal domination, the reader is referred to \cite{KlosMyn14}.  We also point out that eternal domination can also be considered a special case of the Spy Game, where an attacker (spy) moves at speed $s$ on the graph, while the guards are said to ``control'' the spy if one is distance at most $d$ from the spy at the end of their turn (see {\it e.g.}, \cite{CMM^+17,CMNP18}).  Eternal domination is then the special case of the Spy Game with $s = \diam(G)$ and $d=0$.

\subsection{Recent results}

As mentioned, we consider only the ``all guards move'' model.  The cases of paths and cycles for this variant of the game are trivial. In \cite{KM09}, a linear-time algorithm is given to determine $\gamma^{\infty}_{all}(T)$ for all trees $T$. In \cite{BAS15}, the eternal domination game was solved for proper interval graphs. In recent years, significant effort has been made in an attempt to determine the eternal domination number of Cartesian grids, $\gamma^{\infty}_{all}(P_n\square P_m)$ (see Figure~\ref{fig:grids}).  
Exact values were determined for $2\times n$ Cartesian grids in \cite{FinbowMB15, GKM13} and $4\times n$ Cartesian grids in \cite{BeatonFM13}. Bounds for $3\times n$ Cartesian grids were obtained in \cite{FinbowMB15} and improved in \cite{MessingerD15}, and exactly values for all $n$ were recently provided in \cite{finbow2020eternal}.  Bounds for $5\times n$ Cartesian grids were given in \cite{vBommel16}.
For general $m \times n$ Cartesian grids, it is clear that $\gamma^{\infty}_{all}(P_n\square P_m)$ must be at least the domination number, $\gamma(P_n\square P_m)$, and so by the result in \cite{GoncalvesPRT11} it follows that $\gamma^{\infty}_{all}(P_n\square P_m) \geq \lfloor\tfrac{(n-2)(m-2)}{5}\rfloor-4$. The best known upper bound for $\gamma^{\infty}_{all}(P_n\square P_m)$ was determined in \cite{LamprouMS16a}, where it was shown that $\gamma^{\infty}_{all}(P_n\square P_m) \leq \tfrac{mn}{5} + O(m+n)$, thus showing that $\gamma^{\infty}_{all}$ is within $O(m+n)$ of the domination number.

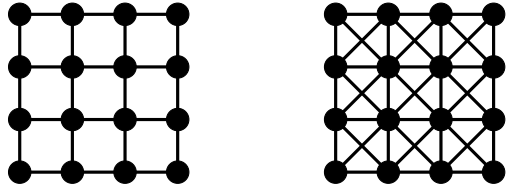
\begin{figure}[h!]\centering

\begin{tikzpicture}[scale=0.7]
\tikzstyle{vertex}=[draw,circle,fill=black,inner sep=3pt]

\draw [very thick] (0,0) grid (3,3);
\draw [very thick] (6,0) grid (9,3);

\draw [very thick] (6,0) to (9,3);
\draw [very thick] (7,0) to (9,2);
\draw [very thick] (8,0) to (9,1);
\draw [very thick] (6,1) to (8,3);
\draw [very thick] (6,2) to (7,3);

\draw [very thick] (6,3) to (9,0);
\draw [very thick] (6,2) to (8,0);
\draw [very thick] (6,1) to (7,0);
\draw [very thick] (7,3) to (9,1);
\draw [very thick] (8,3) to (9,2);

\node[vertex] (1) at (0,0) {};
\node[vertex] (2) at (0,1) {};
\node[vertex] (3) at (0,2) {};
\node[vertex] (4) at (0,3) {};
\node[vertex] (5) at (1,0) {};
\node[vertex] (6) at (1,1) {};
\node[vertex] (7) at (1,2) {};
\node[vertex] (8) at (1,3) {};
\node[vertex] (9) at (2,0) {};
\node[vertex] (10) at (2,1) {};
\node[vertex] (11) at (2,2) {};
\node[vertex] (12) at (2,3) {};
\node[vertex] (13) at (3,0) {};
\node[vertex] (14) at (3,1) {};
\node[vertex] (15) at (3,2) {};
\node[vertex] (16) at (3,3) {};

\node[vertex] (1) at (6,0) {};
\node[vertex] (2) at (6,1) {};
\node[vertex] (3) at (6,2) {};
\node[vertex] (4) at (6,3) {};
\node[vertex] (5) at (7,0) {};
\node[vertex] (6) at (7,1) {};
\node[vertex] (7) at (7,2) {};
\node[vertex] (8) at (7,3) {};
\node[vertex] (9) at (8,0) {};
\node[vertex] (10) at (8,1) {};
\node[vertex] (11) at (8,2) {};
\node[vertex] (12) at (8,3) {};
\node[vertex] (13) at (9,0) {};
\node[vertex] (14) at (9,1) {};
\node[vertex] (15) at (9,2) {};
\node[vertex] (16) at (9,3) {};

\end{tikzpicture}

\caption{The Cartesian grid $P_4\square P_4$ (left) and strong grid $P_4\boxtimes P_4$ (right).}
\label{fig:grids}
\end{figure}

Recently, Finbow et al.~studied the eternal domination game on strong grids, $P_n\boxtimes P_m$, which are, roughly, Cartesian grids where the diagonal edges exist (also known as ``king'' graphs)(see Figure~\ref{fig:grids}). They obtained an upper bound of $\frac{mn}{6}+O(m+n)$ for the eternal domination number of $P_n\boxtimes P_m$ \cite{F16}.  Note that it is trivially known that $\gamma(P_n\boxtimes P_m)=\lceil\frac{m}{3}\rceil \lceil\frac{n}{3}\rceil$. During the preparation of this paper, a parallel work announced the following general lower and upper bounds of $\lfloor\frac{n}{3}\rfloor \lfloor\frac{m}{3}\rfloor+\Omega(n+m)\leq \gamma^{\infty}_{all}(P_n\boxtimes P_m) \leq \lceil\frac{m}{3}\rceil \lceil\frac{n}{3}\rceil + O(m\sqrt{n})$, where $n \leq m$, and thus showing, for large enough values of $n$ and $m$, that $\gamma^{\infty}_{all}(P_n\boxtimes P_m) \approx \gamma(P_n\boxtimes P_m)$ (up to low order terms) \cite{MNP19}.

\subsection{Our results}

We show that $\gamma^{\infty}_{all}(P_n\boxtimes P_m)\leq \frac{mn}{7}+O(m+n)$ for all integers $n,m\geq9$ by adapting the techniques used in \cite{LamprouMS16a}.  
In Section 2, we establish the basic strategy used in the proofs which follow.  It can loosely be thought of as a strategy where the grid is partitioned into subgrids, guards which occupy the corners of $8 \times 8$ grids stay in place, while the guards on the interior of the grid rotate in such a way that a symmetric configuration to the original is obtained.  In Section 3, we show that this strategy easily works for the infinite Cartesian grid, and obtain the main result of the paper in Section 4.  Finally, in Section 5, we compare our results with those reported in \cite{MNP19}. In the spirit of the aforementioned papers focused on ``skinny'' Cartesian grids (those where the smallest dimension is bounded by or is equal to some constant), we show that the strategy presented here gives a better upper bound for $\gamma^{\infty}_{all}(P_n\boxtimes P_m)$ in the case where $n \leq m$ and $n$ is at most some constant.  We also believe that the strategy presented is interesting in its own right and could provide a path for analysis of strong grids in higher dimensions.

\section{Notations and the {\it Alternating} strategy}

We begin by formally defining the graph $P_n\boxtimes P_m$. Let $V(P_n) = \{0,\ldots,n-1\}$ and $V(P_m) = \{0,\ldots,m-1\}$. Then, each vertex in $P_n \boxtimes P_m$ is an ordered pair $(i,j)$, and two vertices $(i_1,j_1)$ and $(i_2,j_2)$ are adjacent if and only if $\max\{|i_2-i_1|,|j_2-i_j|\}=1$ (see Figure~\ref{fig:grids}).

In order to eternally dominate $P_n\boxtimes P_m$, we consider a strategy that cycles through two families of dominating sets, $D$ and $D'$ (see Figure~\ref{fig:domsets}).  Let $D$ be a set of vertices in $P_n\boxtimes P_m$ with the property that if $(i,j)$ is in  $D$ then so are $(i+2,j+1)$ and $(i-1,j+3)$. 
This definition implies that $D$ has a periodic nature, where every seventh vertex in a row or column of $P_n\boxtimes P_m$ contains a vertex in $D$.
Hence, $D$ can be viewed as a dominating set that contains the vertices $(i+2k+7l,j+k+7l)$ and $(i+k+7l,j-3k+7l)$ for some $i,j\geq0$ and all integers $k,l$ such that the resulting vertices have an $x$-coordinate and a $y$-coordinate greater than or equal to zero.  Similarly,
$D'$ is the dominating set that contains the vertices $(i+k+7l,j+3k+7l)$ and $(i+2k+7l,j-k+7l)$ for some $i,j\geq0$ and all integers $k,l$ such that the resulting vertices have an $x$-coordinate and a $y$-coordinate greater than or equal to zero.

If the guards are in a $D$ configuration, then the strategy for the guards is to have one guard move to the attacked vertex and for the rest of the guards to move accordingly to move into a $D'$ configuration and vice versa. 
For most attacks on the interior of the grid, only one response is possible.  However, if the guards occupy a $D$ or $D'$ configuration that contains $(i,j)$, then $(i-1,j)$ and $(i+1,j)$ are adjacent to two guards (assuming $(i,j)$ is far enough from the borders of the grid).  In the case where one of these vertices is attacked, the guard that is diagonally adjacent will defend against the attack (not the guard at $(i,j)$). Due to the guards alternating between two families of configurations $D$ and $D'$, we call this strategy, the {\it Alternating} strategy.

In the {\it Alternating} strategy, there are {\it anchor} guards which do not move from their vertices after an attack and they are determined by which vertex is attacked and the current configuration of the guards. Essentially, the anchor guards occupy the corners of $8\times8$ subgrids inside which the other guards move to protect against attacks and alternate to the next configuration.

\begin{figure}[h!]\centering

\begin{tikzpicture}[scale=0.45]
\tikzstyle{vertex}=[draw,circle,fill=black,inner sep=3pt]

\draw [very thick] (0,0) grid (10,10);
\draw [very thick] (14,0) grid (24,10);

\node[vertex] (1) at (0.5,0.5) {};
\node[vertex] (2) at (2.5,1.5) {};
\node[vertex] (3) at (4.5,2.5) {};
\node[vertex] (4) at (6.5,3.5) {};
\node[vertex] (5) at (8.5,4.5) {};

\node[vertex] (6) at (1.5,4.5) {};
\node[vertex] (7) at (3.5,5.5) {};
\node[vertex] (8) at (5.5,6.5) {};
\node[vertex] (9) at (7.5,7.5) {};
\node[vertex] (10) at (9.5,8.5) {};

\node[vertex] (11) at (7.5,0.5) {};
\node[vertex] (12) at (9.5,1.5) {};

\node[vertex] (13) at (0.5,7.5) {};
\node[vertex] (14) at (2.5,8.5) {};
\node[vertex] (15) at (4.5,9.5) {};

\node at (5,-1.5) {{\LARGE$D$}};

\node[vertex] (1) at (14.5,0.5) {};
\node[vertex] (2) at (15.5,3.5) {};
\node[vertex] (3) at (17.5,2.5) {};
\node[vertex] (4) at (19.5,1.5) {};
\node[vertex] (5) at (21.5,0.5) {};

\node[vertex] (6) at (14.5,7.5) {};
\node[vertex] (7) at (16.5,6.5) {};
\node[vertex] (8) at (18.5,5.5) {};
\node[vertex] (9) at (20.5,4.5) {};
\node[vertex] (10) at (22.5,3.5) {};

\node[vertex] (11) at (17.5,9.5) {};
\node[vertex] (12) at (19.5,8.5) {};
\node[vertex] (13) at (21.5,7.5) {};
\node[vertex] (14) at (23.5,6.5) {};

\node at (19,-1.5) {{\LARGE$D'$}};
\end{tikzpicture}

\caption{Snapshot of a $10\times10$ subgrid of a much larger grid, showing the positions of the guards in a $D$ (left) and a $D'$ (right) configuration.}
\label{fig:domsets}
\end{figure}
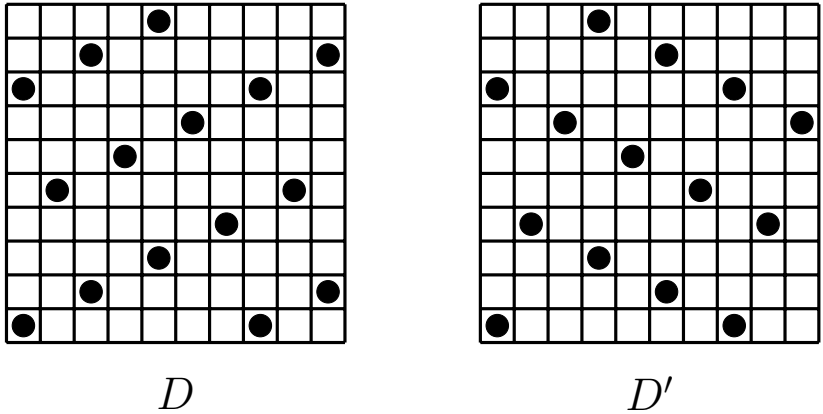

\section{Eternally dominating $P_{\infty} \boxtimes P_{\infty}$}

\begin{theorem}
The {\it Alternating} strategy eternally dominates $P_{\infty} \boxtimes P_{\infty}$.
\end{theorem}

\begin{proof}
Consider the guards initially beginning in a $D$ configuration (where now coordinates are permitted to be any integer).  We will show how the guards can move to a $D'$ configuration containing the attacked vertex for all possible attacks (within symmetry). We omit the proof of the movements of the guards from a $D'$ configuration to a $D$ configuration that contains the attacked vertex as it is analogous to the movements in the opposite direction.

Due to symmetry, we only have to analyse the eight possible attacks on the vertices adjacent to a guard that occupies $(i,j)$. We only consider the movements of the guards in the corresponding $8\times8$ subgrid of the attacked vertex as the remaining subgrids will all be symmetric to this one and so, the movements of the guards as well. Finally, we only have to analyze four of the eight possible attacks since an attack at $(i+1,j)$ is symmetric to an attack at $(i-1,j-1)$, an attack at $(i+1,j+1)$ is symmetric to an attack at $(i-1,j)$, an attack at $(i-1,j+1)$ is symmetric to an attack at $(i+1,j-1)$, and an attack at $(i,j+1)$ is symmetric to an attack at $(i,j-1)$.

\begin{table}[h!]
\centering
\scalebox{0.93}{\begin{tabular}{|c|c|c|}
  \hline 
   Attacked vertex & Anchor vertices & Guard movements \\ \hline
   \multirow{6}{*}{$(i,j-1)$} &  & $(i,j)\rightarrow(i,j-1)$ \\
   & $(i-1,j+3)$  & $(i+1,j-3)\rightarrow(i+2,j-2)$ \\
   & $(i-1,j-4)$  & $(i+3,j-2)\rightarrow(i+4,j-3)$ \\
   & $(i+6,j+3)$ & $(i+5,j-1)\rightarrow(i+5,j)$ \\
   & $(i+6,j-4)$ & $(i+4,j+2)\rightarrow(i+3,j+1)$ \\
   &  & $(i+2,j+1)\rightarrow(i+1,j+2)$ \\ \hline
   \multirow{6}{*}{$(i+1,j-1)$} &  & $(i,j)\rightarrow(i+1,j-1)$ \\
   & $(i-4,j+5)$  & $(i+2,j+1)\rightarrow(i+2,j+2)$ \\
   & $(i-4,j-2)$ & $(i+1,j+4)\rightarrow(i,j+3)$ \\
   & $(i+3,j+5)$ & $(i-1,j+3)\rightarrow(i-2,j+4)$ \\
   & $(i+3,j-2)$ & $(i-3,j+2)\rightarrow(i-3,j+1)$ \\
   &  & $(i-2,j-1)\rightarrow(i-1,j)$ \\ \hline
   \multirow{6}{*}{$(i+1,j)$} &  & $(i,j)\rightarrow(i-1,j+1)$ \\
   & $(i-3,j+2)$  & $(i+2,j+1)\rightarrow(i+1,j)$ \\
   & $(i-3,j-5)$  & $(i+3,j-2)\rightarrow(i+3,j-1)$ \\
   & $(i+4,j+2)$ & $(i+1,j-3)\rightarrow(i+2,j-4)$ \\
   & $(i+4,j-5)$ & $(i-1,j-4)\rightarrow(i,j-3)$ \\
   &  & $(i-2,j-1)\rightarrow(i-2,j-2)$ \\ \hline
   \multirow{6}{*}{$(i+1,j+1)$} &  & $(i,j)\rightarrow(i+1,j+1)$ \\
   & $(i-2,j-1)$  & $(i+2,j+1)\rightarrow(i+3,j)$ \\
   & $(i-2,j+6)$  & $(i+4,j+2)\rightarrow(i+4,j+3)$ \\
   & $(i+5,j-1)$ & $(i+3,j+5)\rightarrow(i+2,j+4)$ \\
   & $(i+5,j+6)$ & $(i+1,j+4)\rightarrow(i,j+5)$ \\
   &  & $(i-1,j+3)\rightarrow(i-1,j+2)$ \\ \hline 
\end{tabular}}
\caption{Movements of guards from a $D$ to a $D'$ configuration in the $8\times8$ subgrid corresponding to all possible attacks (less symmetric cases).}
\label{tab:guardmoves}
\end{table}

\begin{figure}[h!]\centering

\begin{tikzpicture}[scale=0.5]
\tikzstyle{vertex}=[draw,circle,fill=black,inner sep=3pt]

\draw [very thick] (0,0) grid (8,8);
\draw [very thick] (13,0) grid (21,8);
\draw [very thick] (0,-11) grid (8,-3);
\draw [very thick] (13,-11) grid (21,-3);

\draw [->,ultra thick] (1.5,4.5) to (1.5,3.5);
\draw [->,ultra thick] (2.5,1.5) to (3.5,2.5);
\draw [->,ultra thick] (4.5,2.5) to (5.5,1.5);
\draw [->,ultra thick] (6.5,3.5) to (6.5,4.5);
\draw [->,ultra thick] (5.5,6.5) to (4.5,5.5);
\draw [->,ultra thick] (3.5,5.5) to (2.5,6.5);

\node[vertex] (1) at (1.5,4.5) {};
\node[vertex][gray] (2) at (2.5,1.5) {};
\node[vertex][gray] (3) at (4.5,2.5) {};
\node[vertex][gray] (4) at (6.5,3.5) {};
\node[vertex][gray] (5) at (3.5,5.5) {};
\node[vertex][gray] (6) at (5.5,6.5) {};
\node[vertex][gray] (7) at (6.5,3.5) {};

\node[vertex] [gray] at (0.5,0.5) {};
\node[vertex] [gray] at (7.5,0.5) {};
\node[vertex] [gray] at (0.5,7.5) {};
\node[vertex] [gray] at (7.5,7.5) {};

\node at (4,-1.5) {$(i,j-1)$ attacked};

\draw [->,ultra thick] (14.5,4.5) to (14.5,3.5);
\draw [->,ultra thick] (15.5,1.5) to (16.5,2.5);
\draw [->,ultra thick] (17.5,2.5) to (18.5,1.5);
\draw [->,ultra thick] (19.5,3.5) to (19.5,4.5);
\draw [->,ultra thick] (18.5,6.5) to (17.5,5.5);
\draw [->,ultra thick] (16.5,5.5) to (15.5,6.5);

\node[vertex][gray] (1) at (14.5,4.5) {};
\node[vertex][gray] (2) at (15.5,1.5) {};
\node[vertex]       (3) at (17.5,2.5) {};
\node[vertex][gray] (4) at (19.5,3.5) {};
\node[vertex][gray] (5) at (16.5,5.5) {};
\node[vertex][gray] (6) at (18.5,6.5) {};
\node[vertex][gray] (7) at (19.5,3.5) {};

\node[vertex] [gray] at (13.5,0.5) {};
\node[vertex] [gray] at (20.5,0.5) {};
\node[vertex] [gray] at (13.5,7.5) {};
\node[vertex] [gray] at (20.5,7.5) {};

\node at (17,-1.5) {$(i+1,j-1)$ attacked};

\draw [->,ultra thick] (1.5,-6.5) to (1.5,-7.5);
\draw [->,ultra thick] (2.5,-9.5) to (2.5,-8.5);
\draw [->,ultra thick] (4.5,-8.5) to (5.5,-9.5);
\draw [->,ultra thick] (6.5,-7.5) to (6.5,-6.5);
\draw [->,ultra thick] (5.5,-4.5) to (4.5,-5.5);
\draw [->,ultra thick] (3.5,-5.5) to (2.5,-4.5);

\node[vertex][gray] (1) at (1.5,-6.5) {};
\node[vertex][gray] (2) at (2.5,-9.5) {};
\node[vertex][gray] (3) at (4.5,-8.5) {};
\node[vertex][gray] (4) at (6.5,-7.5) {};
\node[vertex]       (5) at (3.5,-5.5) {};
\node[vertex][gray] (6) at (5.5,-4.5) {};
\node[vertex][gray] (7) at (6.5,-7.5) {};

\node[vertex] [gray] at (0.5,-10.5) {};
\node[vertex] [gray] at (7.5,-10.5) {};
\node[vertex] [gray] at (0.5,-3.5) {};
\node[vertex] [gray] at (7.5,-3.5) {};

\node at (4,-12.5) {$(i+1,j)$ attacked};

\draw [->,ultra thick] (14.5,-6.5) to (14.5,-7.5);
\draw [->,ultra thick] (15.5,-9.5) to (16.5,-8.5);
\draw [->,ultra thick] (17.5,-8.5) to (18.5,-9.5);
\draw [->,ultra thick] (19.5,-7.5) to (19.5,-6.5);
\draw [->,ultra thick] (18.5,-4.5) to (17.5,-5.5);
\draw [->,ultra thick] (16.5,-5.5) to (15.5,-4.5);

\node[vertex][gray] (1) at (14.5,-6.5) {};
\node[vertex]       (2) at (15.5,-9.5) {};
\node[vertex][gray] (3) at (17.5,-8.5) {};
\node[vertex][gray] (4) at (19.5,-7.5) {};
\node[vertex][gray] (5) at (16.5,-5.5) {};
\node[vertex][gray] (6) at (18.5,-4.5) {};
\node[vertex][gray] (7) at (19.5,-7.5) {};

\node[vertex] [gray] at (13.5,-10.5) {};
\node[vertex] [gray] at (20.5,-10.5) {};
\node[vertex] [gray] at (13.5,-3.5) {};
\node[vertex] [gray] at (20.5,-3.5) {};

\node at (17,-12.5) {$(i+1,j+1)$ attacked};
\end{tikzpicture}

\caption{Movements of guards from a $D$ to a $D'$ configuration in the $8\times8$ subgrid corresponding to all possible attacks (less symmetric cases). The black guard occupies vertex $(i,j)$ and the four anchor guards are the guards in the corners.}
\label{fig:guardmoves}
\end{figure}
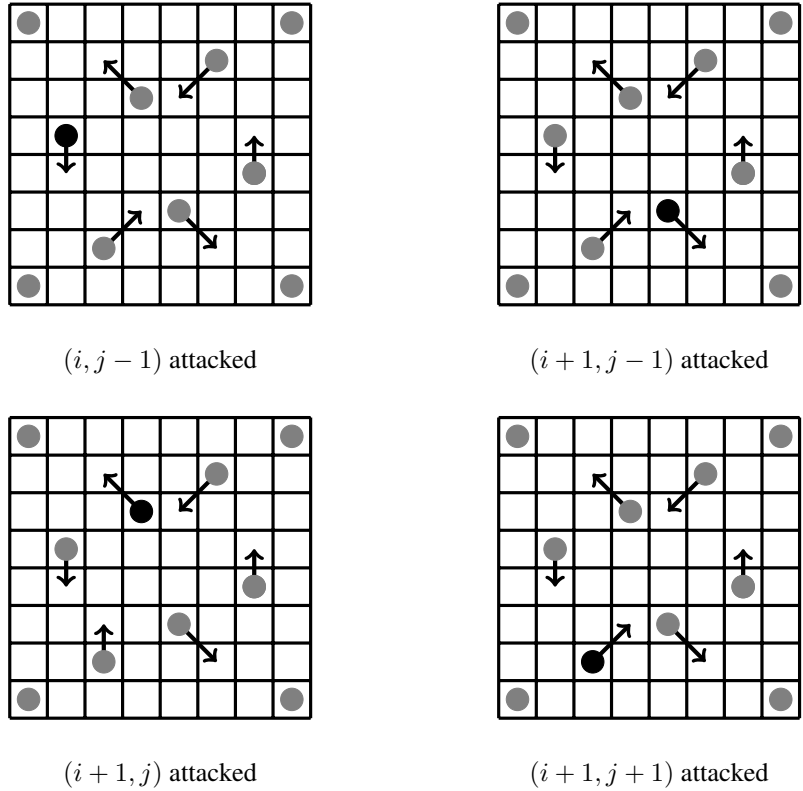

It is easy to verify that the guards' movements are possible and that they transition into a $D'$ configuration after each attack (see Figure~\ref{fig:guardmoves}). Since the grid is infinite, there are an infinite number of guards occupying the vertices of a $D$ or $D'$ configuration and so, any time a guard is required to move to a vertex by the {\it Alternating} strategy, he will always exist and, from Table~\ref{tab:guardmoves} and Figure~\ref{fig:guardmoves}, we know the guards will always transition from $D$ to $D'$ or vice versa with the attacked vertex occupied. Thus, the guards can clearly do this strategy indefinitely and hence, they eternally dominate $P_{\infty} \boxtimes P_{\infty}$.
\end{proof} 

\section{Eternally dominating $P_n \boxtimes P_m$}

We proceed to the case where the grid is finite and show that for $n,m\geq9$, $\gamma^\infty_{all}(P_m \boxtimes P_{n})\leq \frac{mn}{7}+O(m+n)$. In order to facilitate obtaining an exact value for the $O(m+n)$ term, we consider different cases which depend on the divisibility of $n$ and $m$.  We first provide a strategy for the finite grid $P_n \boxtimes P_m$ when $n\equiv m \equiv 2 \ (\mathrm{mod} \ 7)$ and $n,m\geq9$, which utilizes the {\it Alternating} strategy with an adjustment to deal with the borders of the grid.
We then generalize this strategy to any $n\times m$ grid for $n,m\geq9$ by employing two disjoint strategies.

\begin{theorem} \label{thm:main}
For any two integers $n,m\geq9$ such that $n\equiv m \equiv 2 \ (\mathrm{mod} \ 7)$, $\gamma^\infty_{all}(P_m \boxtimes P_n)\leq \frac{mn}{7}+\frac{8}{7}(m+n-1)$.
\end{theorem}

\begin{proof}
We use the fact that $n\equiv m \equiv 2 \ (\mathrm{mod} \ 7)$ to reduce the analysis of the guards' strategy to the case of a $9\times9$ grid. Essentially, the non-border vertices can be partitioned into $\frac{(n-2)(m-2)}{7}$ $7\times7$ subgrids since $n\equiv m \equiv 2 \ (\mathrm{mod} \ 7)$. We place one guard in each of the corners of the $n\times m$ grid and these guards never move. Finally, we can partition the sides of the grid (not including the corners) into paths of seven vertices.

We implement the {\it Alternating} strategy in all of the $7\times7$ subgrids which means they will all have identical configurations. Hence, we can focus just on the case of the $7\times7$ subgrids that touch the border vertices of the grid to ensure that the guards from the border can move into these grids when needed. Thus, we contract the $n\times m$ grid into a $9\times 9$ grid and show a winning strategy for the guards there which ensures the borders of the $n\times m$ grid will be protected symmetrically for each of the paths of seven vertices that make up the borders and that the $7\times 7$ subgrids adjacent to the borders are symmetric to all the other $7\times 7$ subgrids. This strategy can then be easily ``translated'' to any of the $7\times 7$ subgrids that touch the border vertices to gain a global strategy.

We show a winning strategy for the guards in the $9\times 9$ grid where four guards remain in the corners indefinitely, five guards occupy each of the paths of seven vertices in between the corners (on the borders of the grid), and seven guards from the {\it Alternating} strategy occupy the $7\times7$ subgrid in the middle (see Figure~\ref{fig:9x9grid}). The five guards on each of the paths of seven vertices initially occupy the five central vertices, leaving the leaves empty. If any border vertices get attacked, then they must be one of the leaves of the paths of seven vertices and the closest guard on the corresponding path moves to the attacked vertex. The remaining four guards on the same path stay still, as well as all seven of the guards in the interior of the grid, and the guards on each of the other paths move to a symmetric formation as the path that was attacked. Any subsequent attack on a border vertex is dealt with in the same fashion, {\it i.e.}, if the other leaf is attacked, then the guards on the path move into a symmetric formation with one guard on the attacked leaf and the other four guards occupying a sequence of four vertices non-adjacent to the fifth guard and not including any leaf. If an attack occurs on a non-leaf vertex of the path, then the five guards move back into their initial formation which includes neither of the leaves. Note that the interior guards never move if a border vertex is attacked and the guards on each of the paths are in symmetric positions.

Now, for each guard the {\it Alternating} strategy requires to move in from a border vertex, it requires a guard to move out from the interior vertices. The exchange is easy to facilitate since the guard moving out of the interior will always move onto the same border path that the guard moving in to the interior previously occupied. In all three of the possible configurations of the guards on the border vertices, the guards occupy a dominating set of the row or column of vertices adjacent to them in the $7\times7$ subgrid. Hence, there is a guard available to move to whichever vertex requires a guard to move to it and the guard leaving the interior can always move onto the border path as the guards can easily maneuver to leave an adjacent vertex empty for him while maintaining one of the three formations.

Thus, the {\it Alternating} strategy with the extra guards on the borders of the grid, eternally dominates $P_m \boxtimes P_n$. This strategy uses $\frac{(n-2)(m-2)}{7}+(2)\frac{5}{7}(m-2+n-2)+4=\frac{mn}{7}+\frac{8}{7}(m+n-1)$ guards which gives our result.
\end{proof}

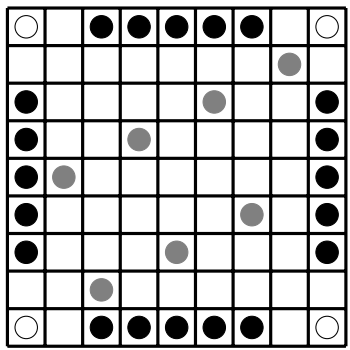
\begin{figure}[h!]\centering

\begin{tikzpicture}[scale=0.5]
\tikzstyle{vertex}=[draw,circle,fill=white,inner sep=3pt]

\draw [very thick] (0,0) grid (9,9);

\node[vertex] (1) at (0.5,8.5) {};
\node[vertex] (2) at (8.5,8.5) {};
\node[vertex] (3) at (0.5,0.5) {};
\node[vertex] (4) at (8.5,0.5) {};

\node[vertex][black] (5) at (0.5,6.5) {};
\node[vertex][black] (6) at (0.5,5.5) {};
\node[vertex][black] (7) at (0.5,4.5) {};
\node[vertex][black] (8) at (0.5,3.5) {};
\node[vertex][black] (9) at (0.5,2.5) {};

\node[vertex][black] (10) at (8.5,6.5) {};
\node[vertex][black] (11) at (8.5,5.5) {};
\node[vertex][black] (12) at (8.5,4.5) {};
\node[vertex][black] (13) at (8.5,3.5) {};
\node[vertex][black] (14) at (8.5,2.5) {};

\node[vertex][black] (15) at (2.5,0.5) {};
\node[vertex][black] (16) at (3.5,0.5) {};
\node[vertex][black] (17) at (4.5,0.5) {};
\node[vertex][black] (18) at (5.5,0.5) {};
\node[vertex][black] (19) at (6.5,0.5) {};

\node[vertex][black] (20) at (2.5,8.5) {};
\node[vertex][black] (21) at (3.5,8.5) {};
\node[vertex][black] (22) at (4.5,8.5) {};
\node[vertex][black] (23) at (5.5,8.5) {};
\node[vertex][black] (24) at (6.5,8.5) {};

\node[vertex][gray] (25) at (2.5,1.5) {};
\node[vertex][gray] (26) at (4.5,2.5) {};
\node[vertex][gray] (27) at (6.5,3.5) {};
\node[vertex][gray] (28) at (1.5,4.5) {};
\node[vertex][gray] (29) at (3.5,5.5) {};
\node[vertex][gray] (30) at (5.5,6.5) {};
\node[vertex][gray] (31) at (7.5,7.5) {};
\end{tikzpicture}

\caption{One possible configuration of the guards in the $9\times 9$ strong grid with the guards in white in the corners never moving, the guards on the paths of length seven between the corners in black, and the remaining guards in gray.}
\label{fig:9x9grid}
\end{figure}

We now use Theorem~\ref{thm:main} to prove $\gamma^\infty_{all}(P_m \boxtimes P_{n})\leq \frac{mn}{7}+O(m+n)$ for grids in general when $n,m\geq9$ and to give exact values of $O(m+n)$ in these bounds by employing two disjoint strategies as follows. The strategy from Theorem~\ref{thm:main} is used for the largest $a\times b$ subgrid in the $n\times m$ grid such that $a\equiv b \equiv 2 \ (\mathrm{mod} \ 7)$ and a separate strategy is used for the remaining unguarded vertices where none of the guards from the two strategies are ever utilised in the other strategy and never occupy a vertex that the other strategy is responsible for protecting.

\begin{theorem}\label{thm:allcases}
For any two integers $n,m\geq9$, $\gamma^\infty_{all}(P_m \boxtimes P_n)\leq \frac{ab}{7}+\frac{8}{7}(a+b-1)+\alpha \lceil\frac{n}{2}\rceil+\beta \lceil\frac{m}{2}\rceil - \alpha\beta$ where $a\equiv b \equiv 2 \ (\mathrm{mod} \ 7)$, $0\leq n-a\leq 6$, $0\leq m-b\leq 6$, and
\begin{displaymath}
   \alpha = \left\{
     \begin{array}{ll}
       0 & if\quad m\mod7=2\\
       1 & if\quad m\mod7\in\{3,4\}\\
       2 & if\quad m\mod7\in\{5,6\}\\
       3 & if\quad m\mod7\in\{0,1\}
     \end{array}
   \right.
\end{displaymath}
\begin{displaymath}
   \beta = \left\{
     \begin{array}{ll}
       0 & if\quad n\mod7=2\\
       1 & if\quad n\mod7\in\{3,4\}\\
       2 & if\quad n\mod7\in\{5,6\}\\
       3 & if\quad n\mod7\in\{0,1\}
     \end{array}
   \right.
\end{displaymath}
\end{theorem}

\begin{proof}
The $\frac{ab}{7}+\frac{8}{7}(a+b-1)$ guards follow the strategy in the proof of Theorem~\ref{thm:main} in the $a\times b$ subgrid which will include at least one corner of the $n\times m$ grid. Separately, there remain $n-a$ columns and $m-b$ rows to protect which are all found on the same side of the $n\times m$ grid due to the placement of the $a\times b$ grid. That is, there are $n-a$ consecutive remaining columns and $m-b$ consecutive remaining rows which overlap near one corner of the $n\times m$ grid (see Figure~\ref{fig:abgrid}).  We can guard the $m-b$ remaining rows with $\alpha \lceil\frac{n}{2}\rceil$ guards, since
one guard every two vertices can protect two rows since the two rows are partitioned into disjoint $K_4$ (plus some remainder due to divisibility) and one guard is assigned to each $K_4$ which clearly he can protect. Thus, we use $\lceil\frac{n}{2}\rceil$ guards for every two rows that remain and if there are an odd number of rows remaining, then we use $\lceil\frac{n}{2}\rceil$ guards to protect the last remaining row. Similarly, $\beta \lceil\frac{m}{2}\rceil$ corresponds to the number of guards required to protect the $n-a$ remaining columns.
Since we have over-counted by $\alpha\beta$ overlapping guards, the bound follows.
\end{proof}

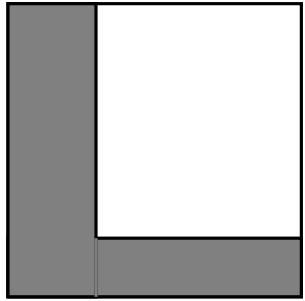
\begin{figure}[h]\centering

\begin{tikzpicture}[scale=0.39]
\draw [very thick] (0,0) rectangle (10,10);
\filldraw[fill=gray, draw=black, very thick] (0,0) rectangle (10,2);
\filldraw[fill=gray, draw=black, very thick] (0,0) rectangle (3,10);
\draw [very thick, gray] (3,0) to (3,2);
\end{tikzpicture}

\caption{The $n\times m$ strong grid with the area in white representing the $a\times b$ subgrid and the area in gray representing the remaining rows and columns.}
\label{fig:abgrid}
\end{figure}

\newpage

Now, we can prove our main result.

\begin{corollary}
For any two integers $n,m\geq9$, $\gamma^\infty_{all}(P_m \boxtimes P_{n})\leq \frac{mn}{7}+O(m+n)$.
\end{corollary}

\begin{proof}
This follows directly from Theorem~\ref{thm:allcases}.
\end{proof}



\section{Comparison with other bounds}

In a parallel work, Mc Inerney, Nisse, and P\'erennes \cite{MNP19} show that if $m \geq n$, then $\gamma^{\infty}_{all}(P_n\boxtimes P_m) \leq \lceil\frac{m}{3}\rceil \lceil\frac{n}{3}\rceil + O(m\sqrt{n})$.
The general configuration which is maintained is to (a) fill some number of rows and columns with stationary guards so that the dimensions of the remaining $m^* \times n^*$ grid satisfy necessary divisibility conditions, (b) add additional rows and columns of guards to allow passage of guards around the outside of the subgrid, (c) partition the subgrid into $m^* \times k$ smaller subgrids, (d) place guards along boundary layers of each subgrid, and (e) place one guard for every nine vertices of the interior of each subgrid in such a way that every attack has a response, transferring guards as need be through the boundary layers.


In the best case ({\it i.e.}, it is not necessary to fill some rows/columns with guards to ensure the remaining subgrid satisfies divisibility conditions), the bound from \cite{MNP19} is
\begin{align}\label{ninth}
\left(\frac{m - 2\alpha}{3\alpha-1}\right)\left(\left\lceil \frac{(n-2)(3\alpha - 3)}{9} \right\rceil + 2n + 6\alpha - 6 \right) + 2\alpha(m+n-2\alpha).
\end{align}
where $\alpha = \frac{k-2}{3} + 1$ and $k$ is the greatest integer less or equal to $\sqrt{n}$ for which $k \equiv 2 \pmod 3$ (note that $k > \sqrt{n}-3$).


The worst case for our bound is when $\alpha = \beta = 3$, $a = n-5$, and $b=m-5$, as this requires packing the most stationary guards around two sides of the grid.  This gives
\begin{align}\label{seventh}
\frac{(m-5)(n-5)}{7}+\frac{8(m+n-11)}{7}+ 3\left\lceil\frac{m}{2}\right\rceil + 3\left\lceil\frac{n}{2}\right\rceil - 9.
\end{align}

\begin{figure}[h!]
    \centering
    \includegraphics[width=3in]{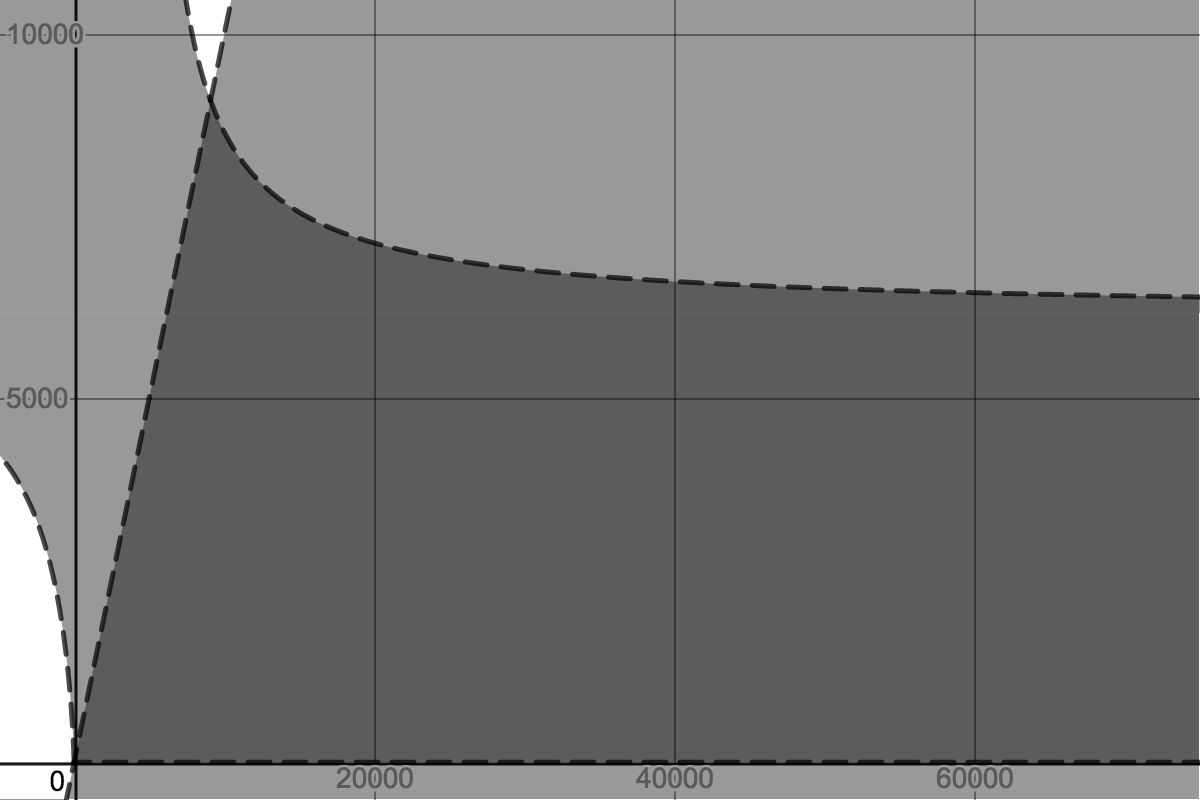}
    \caption{Comparison of bounds (\ref{ninth}) and (\ref{seventh}); $m$ as horizontal axis, $n$ as vertical axis}
    \label{comparison}
\end{figure}

The dark shaded region shown in the graph in Figure \ref{comparison} gives the values of $m$ and $n$ for which the bound in (\ref{seventh}) bests the bound in (\ref{ninth}), using $k = \sqrt{n}-3$ to express (\ref{ninth}) as a function of $m$ and $n$ only.  Note that the bounding function as $m \to \infty$ eventually stays strictly between $n = 6179$ and $n = 6180$, and thus our result is best when $n < 6180$.

We point out that the authors of \cite{MNP19} did not attempt to optimize the constants in their result (nor did we in this paper), only to show that the domination number plus some low order terms was an upper bound for $\gamma^\infty_{all}$.  However, the ``dense'' guards surrounding each subgrid is an integral part of their argument, leading to the $O(m\sqrt{n})$ term in their result which cannot be dropped (unless a new method is found to guard the boundaries).  As a result, even with optimization of constants, our strategy should be preferred for sufficiently ``skinny'' $n \times m$ strong grids.

\section{Acknowledgements}
This work was undertaken while the second, third, and fourth authors were affiliated with Dawson College,and while the eighth author was affiliated with Universit\'e de Montr\'eal.  The sixth author acknowledges the generous support of MITACS and the Globalink Program.  The fourth and seventh authors received financial support for this research from the Fonds de recherche du Québec - Nature et technologies.

\bibliographystyle{alpha}
{\footnotesize \bibliography{eterdom_strong_grid}}

\end{document}